\documentclass[10pt,conference]{IEEEtran}  
\usepackage{fancyhdr}
 \usepackage{amsmath}
\usepackage{cases}
\usepackage{stfloats}
\usepackage{enumerate}
\usepackage{url}
\usepackage{graphicx}
\usepackage{balance}
\usepackage{algorithm}
\usepackage{algpseudocode}

 \usepackage{etoolbox}

\newtoggle{ACM}
\toggletrue{ACM}
\togglefalse{ACM}

\newtoggle{IEEEcls}
\toggletrue{IEEEcls}

\floatname{algorithm}{Algorithm}

\iftoggle{IEEEcls}{
\usepackage{amsthm}
\theoremstyle{plain}
\newtheorem{theorem}{Theorem}

\theoremstyle{definition}
\newtheorem{definition}{Definition}

}

\newif\ifNotUse  
 \NotUsetrue
 
 \newif\ifNotUse  
 \NotUsetrue

 \iftoggle{ACM}{
\setcopyright{none}  
\acmDOI{}
\acmPrice{}
\acmISBN{}
}

\begin{document}
\title{GPU based Real-time Super Hosts Detection at Distributed Edge Routers}
\author{\IEEEauthorblockN{Jie Xu
                          ,Wei Ding
                          ,Xiaoyan Hu
                                  }
\IEEEauthorblockA{
School of Computer Science and Engineering\\ 
Southeast University \\
Nanjing, China\\
Email: xujieip@163.com} 

}

\maketitle

\begin{abstract}
The super host is a special host on the network which contacts with many other hosts during a certain time window. They play important roles in network researches such as scanners detection, resource allocation, spam filtering and so on. How to find super hosts in real time is the foundation of these applications. In this paper, a novel algorithm, denoted as CBAA, is proposed to solve this problem at edge routers. CBAA divides network traffic into different parts. A cube of bits array is devised to store hosts’ linking information of different traffic parts when scanning packets. At the end of each time window, CBAA restores super hosts very fast because there are only a fraction of super hosts in each traffic part. CBAA is also a parallel algorithm. It’s easy to deploy CBAA in GPU to deal with high-speed network traffic in real time. Experiments on a real-world core network prove the advantage of our algorithm.
\end{abstract} 
 

\section{Introduction}
Host's cardinality is one of the most important network attributes, which means the number of other hosts communicated with it during a time window. Host, in this paper, represents a computer(or virtual machine, such as a cloud server) in a network with a unique IP address. 

A super host is a host whose cardinality is larger than a threshold during a time window. Although super hosts take a small part of the hosts (no more than 0.1\%), they play an important role in the network. Super host detection has a wide application in the field of network management and security such as DDoS detection, scanners location, instruction detection and so on.

Being able to detect super hosts at core network's edge routers(e.g., 7000 Gb/s \cite{Report:ChinaNet}) precisely is a burden because updating hosts' states for every coming packet in real time requires high-speed processors and memory (such as high-frequency CPU cores and SRAM). What's more, traffic of recent backbone network passed through several edge routers for burden balance or security reasons. Multi-edge routers force us to detect super host in distributed servers because it's too expensive to collect all packets into a global traffic. In a distributed environment, packets will be processed at different local servers and necessary data be sent to a global server for super hosts detection. A small data structure will reduce the communication time between local servers and global server. 

Most of the previous researchers focused on how to mining super hosts on a local server\cite{HSD:samplingCountingGeometricSteams} \cite{HSD:sampleFlowDistributionEstimate} \cite{HSD:sampleAdaptivePacketSamplForFlowMeasurement}. Fortunately, some of these algorithms can be applied in a distributed situation after modifying.
Because small memory requirement is a goal of previous super host detection algorithms and small data structure size means littler communication time in the distributed environment. 

For the high-speed network, if we want to deal with every packet in time, the processing speed should be as fast as the link speed. If the processing speed is slower than packets arriving speed, a larger packets buffer will be used or some packets will be lost. But previous algorithms, which can be modified for the distributed application, cannot reach nowadays line speed because they try to improve speed only depending on high-speed memory and ignore the potential computing complexity. 
Parallel computing, such as GPGPU, would be a good solution to this problem\cite{GPU:WireSpeedNameLookup}\cite{PC2016:AHighThroughputDPIEngineOnGPUviaAlgorithmImplementationCoOptimization}\cite{PC2015:AHolisticApproachToBuildRealTimeStreamProcessingSystemWithGPU}\cite{PC2017:ExploitingHeterogeneityOfCommunicationChannelsForEfficientGPUSelectionOnMultiGPUNodes}.

Graphics processing unit (GPU) is a specialized device on a computer. As suggested by the name, its intent is to accelerate the graphic process. GPU contains huge operating unit and has more memory controller than CPU. This means GPU can deal with several data parallel and access memory with smaller total latency\cite{GPU:OptimizationPrinciplesApplicationPerformanceCUDA} \cite{GPU:ImplementationsJacobiAlgorithm}. It is more attractive than CPUs for high-speed traffic super hosts detection. 

In this paper, we research the performance of GPU on super hosts detection and proposed a faster and more accuracy algorithm which can be deployed on distributed environment. We make the following contributions in this paper:
\begin{enumerate}
\item   We propose a novel paradigm to detect super hosts in high-speed network by GPU. On a common GPU, super hosts hiding in huge traffic of core network could be detected in real time by our paradigm.

\item	We designed a fast and memory efficient algorithm to detect super hosts. The novel algorithm has a faster speed and higher overall accuracy. Our novel algorithm can be deployed in a distributed environment and its super hosts detection speed is the fastest among previous algorithms when running on the same platform.

\item	We implement our novel algorithm and previous ones on GPU by our paradigm. With our novel algorithm, a desktop GPU can detect super hosts hiding in a real world core network which has throughput as high as 40 Gb/s.
\end{enumerate}

The remainder of the paper is organized as follows: next section provides background information on super hosts detection and related approaches. In section 3, we describe our novel super host's detection algorithm. We explain how to implement super hosts detection on GPU in section 4. Experiments are presented in section 5. We make a conclusion in the last section. 

\section{Related work}
Super hosts detection is a hot topic in network management and network security \cite{CDE:TrackCardinalityDistribution} \cite{HSD:identifyHighCardinalityHosts}.

Chen et al\cite{HSD:LineSpeedAccurateSuperspreaderIdentificationDynamicErrorCompensation} used a bit array to record new flows and a counting bloom filter to estimate hosts' cardinality. The accuracy of this algorithm depends on the new flow identifying method. At a single node, a bit array could make sure that a flow update counting bloom filter at most once. But in the distributed environment, a flow may appear in several routers which caused the counting bloom filter containing error information. Chen claimed that the algorithm could process 2 million packets per second on SRAM. But this speed was still too slow to handle high-speed network traffic.

Wang et al \cite{HSD:ADataStreamingMethodMonitorHostConnectionDegreeHighSpeed} proposed a super host detection algorithm based on a novel structure called double connection degree sketch (DCDS). This algorithm contained two kinds of bit arrays which were used for hosts recovering and high cardinality estimation. DCDS adopted Chinese Remainder Theorem (CRT) to recover HCHs from bits array smartly. The implementation of CRT required large computation operators which limited its speed. 

Liu et al\cite{HSD:DetectionSuperpointsVectorBloomFilter} proposed a novel data streaming method called VBFA. The superior performance of VBFA came from a new data structure, called vector bloom filter (VBF), which was a variant of the standard bloom filter. VBF consisted of six hash functions, four of which take some consecutive bits from the input string as the corresponding value. The information of super hosts was obtained by using the overlapping of hash bit strings of the VBF. 
VBFA viewed each element of VBF as a bit vector \cite{DC:aLinearTimeProbabilisticCountingDatabaseApp} and estimated host's cardinality from these bit vectors. VBFA was faster than DCDS when scanning packets, but it would consume much more time when recovering super hosts if there were many super hosts hiding in the traffic because it generating more candidate IPs than DCDS.

Both DCDS and VBF tried to speed up packets processing by using SRAM instead of DRAM. Although SRAM is faster, their size is small. What's more, the processing ability of a central processing unit (CPU) is limited.

General-purpose computing on GPU(GPGPU) is the use of GPU, which typically handles computation only for computer graphics, to perform computation in applications handled by CPU traditionally\cite{GPU:ASurveyofGeneral-PurposeComputationOnGraphicsHardware}. Seon-Ho Shin et al \cite{HSD:GPU:2014:AGrandSpreadEstimatorUsingGPU} proposed a GPU based super hosts detection algorithm called GSE. GSE firstly used GPU to estimate host's cardinality with a novel data structure called compact spread estimator (CSE). CSE consisted of two components: one for storing flows, and the other for cardinality estimating. For every packet, CSE checked it in a Collision-tolerant hash table (CTH). If flow corresponding with this packet had not appeared before, CSE would set a bit in CSE. But CSE could only record the cardinality information and it couldn't recover super hosts. In another word, CSE needs to maintain an IP list to record the appearing IP in a time window. In distributed case, the transmission of this IP list between different nodes would reduce the global efficiency.

DCDS, VBFA and GSE can be modified to run in the distributed situation. But DCDS and VBFA have a lower processing speed, GSE consumes too much memory.

Unlike the previous algorithm, our novel algorithm can recover super hosts from a smart structure with a high-speed and occupies a small memory. Next section we will introduce our algorithm in detail.

\section{Distributed super hosts detection}
Suppose there are two networks, $NI$ and $NO$, where $NI$ is the inner network that we want to monitor and $NO$ is the outer network which contacts with $NI$ through a set of edge routers $ER$(for instance, $NO$ is the Internet and $NI$ may be an autonomous system managed by an ISP and $ER$ is a border router running BGP \cite{BGP:ASurveyBGPSecurityIssuesSolutions}\cite{BGP:SecureBorderGatewayProtocol}). All of these packets between $NI$ and $NO$ will be relayed by $ERs$. In the view of $ER$, every packet includes an IP address pair $<innerip,outerip>$, where $innerip$ represents an IP address in $NI$ and $outerip$ is an IP address belong to $NO$. Let $PKT_i$ represent the set of packets passing from the $i$th $ER_i$ in a time window and denote $|PKT_i|$ as the number of packets during this time window. Fig.\ref{Super_hosts_detection_model} illustrates how to detect super hosts of $NI$. 

One of the simplest ways is sending IP pairs of packets at different routers to a global server and mining super hosts at it. But these edge routers may distribute at different places and it's too expensive to transmit gigabits of IP pairs every second. A reasonable method is to scan packets on local servers near edge routers and only send small necessary information to global memory. 

\begin{figure}[!ht]
\centering
\includegraphics[width=0.47\textwidth]{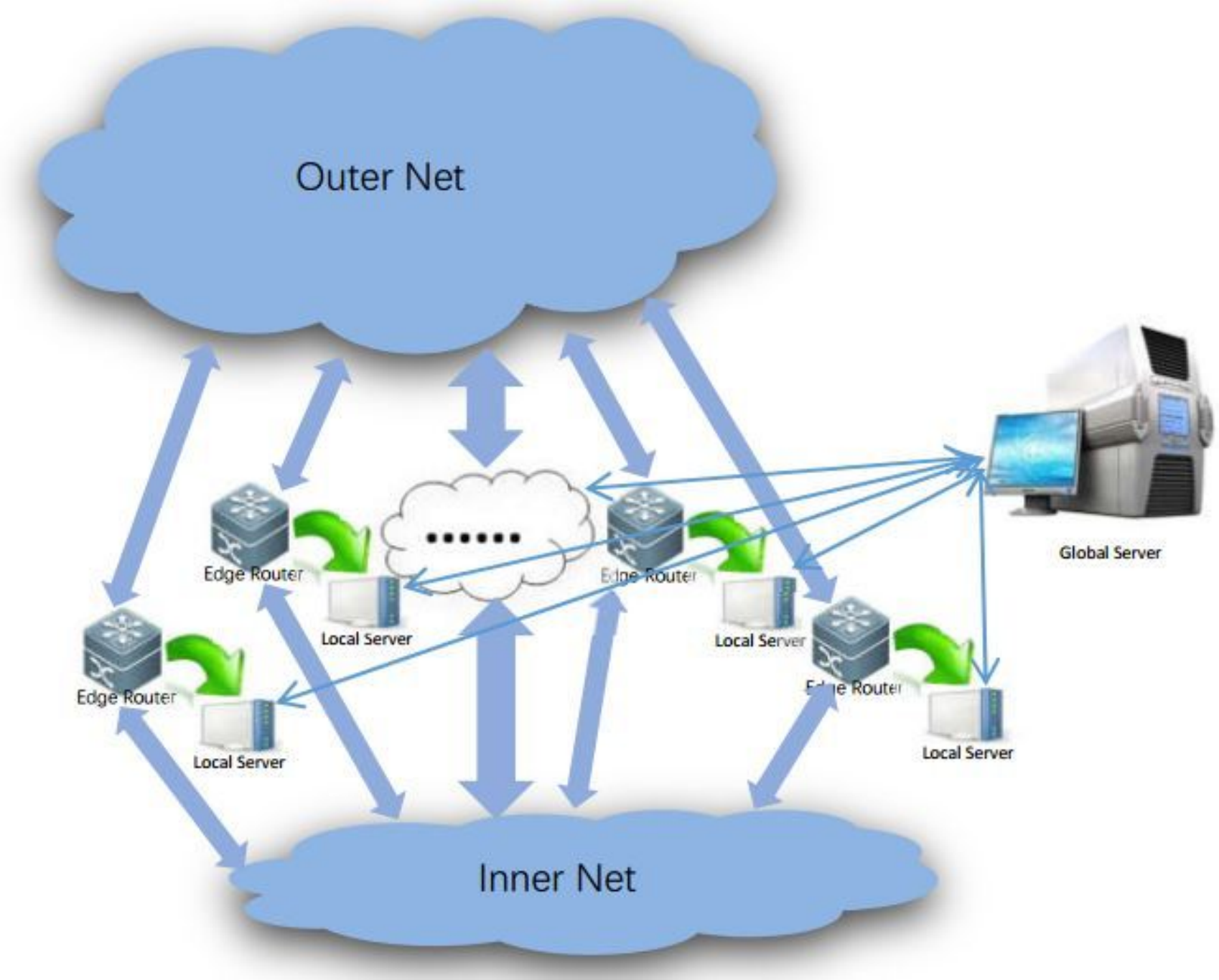}
\caption{Super hosts detection and monitor}
\label{Super_hosts_detection_model}
\end{figure}

We denote $IIP$(Inner IPs) as the set of IP addresses belong to $NI$ appearing in a time window and $OIP$(Outer IPs) as the set of $NO$'s IPs addresses contacted with hosts in $IIP$ during this time window. For $iip \in IIP$, let $OP(iip)$ represent the opposite IPs (belong to $OIP$ and communicate with $iip$ in this time window) set of $iip$. The number of opposite IPs $\left|OP(iip)\right|$ is the cardinality of $iip$. All these packets with the same inner and outer IP addresses generate a flow. Let $FLW$ represent the set of flows. The number of flow $|FLW|$ in a time window is the sum of each inner IP's cardinality and $|FLW|=\sum_{iip \in IIP}|OP(iip)|$.

 Super host's definition used in this paper is shown below.
\begin{definition}[Super Host]
\label{def-hothosts}
Given a threshold $\theta$, a super host with IP address $iip$ in monitored network $NI$ is the host that has no less $\theta$ IP addresses in outer network $NO$ sending packets to or receiving packets from $iip$ during a certain time (called time window). Briefly, if $|OP(iip)|\geq \theta$, iip is a super host.
\end{definition}

Our distributed super hosts detection algorithm contains two parts: scanning packets on a local server and recovering super hosts on the global server.

Local server of an edge router $ER$ will record these hosts' cardinalities passed through $ER$. How to record cardinalities efficiently is a key step. A smart data structure proposed in this paper is designed to solve this problem.

\subsection{Hosts cardinalities estimating}
The precision of super host detection depends on hosts' cardinality estimator. 
For an IP address $iip \in IIP$, let $Pkt(iip)$ represent the set of packets in certain time window whose source or destination address is $iip$. When a packet in $Pkt(iip)$ is sent from $iip$ to IP in outer network, $iip$ will be the source address; when the packet is sent to $iip$ from outer network, $iip$ will be the destination address in this packet. Ignoring the direction of these packets, we can extract their IP pairs $IPpair(iip)$. $IPpair(iip)$ has the same number of elements as $Pkt(iip)$ does. Every element in $IPpair(iip)$ is an IP pair like $<iip,oip>$ where $oip \in OIP$. Suppose there are $k$ packets in $Pkt(iip)$. Go a step further, we can reduce $IPpair(iip)$ to a set of outer network's IPs $OutIP(iip)=\{oip_1,oip_2,\cdots ,oip_k\}$ where $oip_i \in OP(iip)$.
Because $iip$ could send several packets to another host in outer network, so an outer IP $oip$ may appear several times in $OutIP(iip)$ and $k \geq |OP(iip)|$. The task of calculating $iip$'s cardinality is to get the number of distinct elements in $OutIP(iip)$($|OP(iip)|$) by scanning $OutIP(iip)$ once.

An exact way to get the cardinality of host $iip$ is to store every distinct element of $OutIP(iip)$ in memory with a data structure such as list, hash table or red-black tree and so on. And at last, calculate the number of elements in memory which is $iip$'s cardinality. 

For example, if we use list to store distinct IPs in $OutIP(iip)$, $4*|OP(iip)|$ bytes of memory is required. When scanning $oip \in OutIP(iip)$, we should compare it with each element in list one by one. If $oip$ already appears in this list we will scan another IP. Otherwise, add $oip$ to this list. The time complexity of scanning an IP is $O(|OP(iip)|)$. So the total time complexity of scanning $OutIP(iip)$ is $O(k*|OP(iip)|)$. Although this method can get the precise answer, it requires too much processing time and memory. What's more, this method could not be implemented parallel. 

Several rough estimation methods \cite{DC:AnOptimalAlgorithmDistinctElementProblem} \cite{DC:aLinearTimeProbabilisticCountingDatabaseApp} were proposed to get an estimation of host cardinality with a little deviation but occupied a fixed size of memory which was much smaller than the exact one needed.  OPT\cite{DC:AnOptimalAlgorithmDistinctElementProblem} is a memory efficient algorithm to estimate host's cardinality, but its process is a little complex. Bits vector estimator(BVE)\cite{DC:aLinearTimeProbabilisticCountingDatabaseApp} only update a bit when scanning an element. We choose BVE as our basic cardinality estimator method.

Before scanning $OutIP(iip)$, we initialize $g$ bits to zero. These $g$ bits turn up to a bits vector. Every IP in $OutIP(iip)$ will be mapped to a bit in this bits vector by a hash function and set the bit to 1. After scanning all these IPs, we can get $iip$'s cardinality estimation by the following equation \cite{DC:aLinearTimeProbabilisticCountingDatabaseApp}: 
\begin{equation}\label{BEV_est_value}
\widehat{|OP(iip)|}=-g*ln(\frac{z}{g})
\end{equation} 
Where $|OP(iip)|'$ is $iip$'s cardinality estimating value, $z$ is the number of zero bits in bits vector.

In practice, the bigger $g$ is, the more accurate $|OP(iip)|'$ will be. But in the network, most hosts' cardinalities are very small. It's unwise to allocate $g$ bits for every inner IPs. Consequently, we put several bits vector together to record several inner IPs' cardinalities and we call these several bits vector as bits array (BA). A bits vector is a column in BA. An inner IP $iip$ will be projected to a bits vector in BA by a hash function and this bits vector will be used to record $iip$'s cardinality. A single BA could be used to record several inner IPs' cardinality, but it has the following two weakness:
\begin{enumerate}
\item A single BA can not recover super hosts directly. $iip$ will be projected to a column in BA by a one-way hash function. For a column in BA, we couldn't determine which $iip$ is projected to it.

\item The estimating value would be higher than the real cardinality. Because every column of a BA could be used to record several inner IPs' cardinalities at the same time window. Some bits of these columns could be set by different inner IPs. If didn't remove this influence, we will get error result.
\end{enumerate}

In order to solve these problems, we use several bits arrays together and modify the estimating equation by reducing some zero bits number.

We separate bits arrays into two categories: restoring arrays(RAs) and validating arrays(VAs). RAs are used to restore super hosts and reducing the influence of bits sharing. VAs only have the function of reducing bits sharing influence. The set of RAs and VAs is called cardinality sketch(CS) as shown in fig.\ref{BitsArrayCube}.

\begin{figure}[!ht]
\centering 
\includegraphics[width=0.47\textwidth]{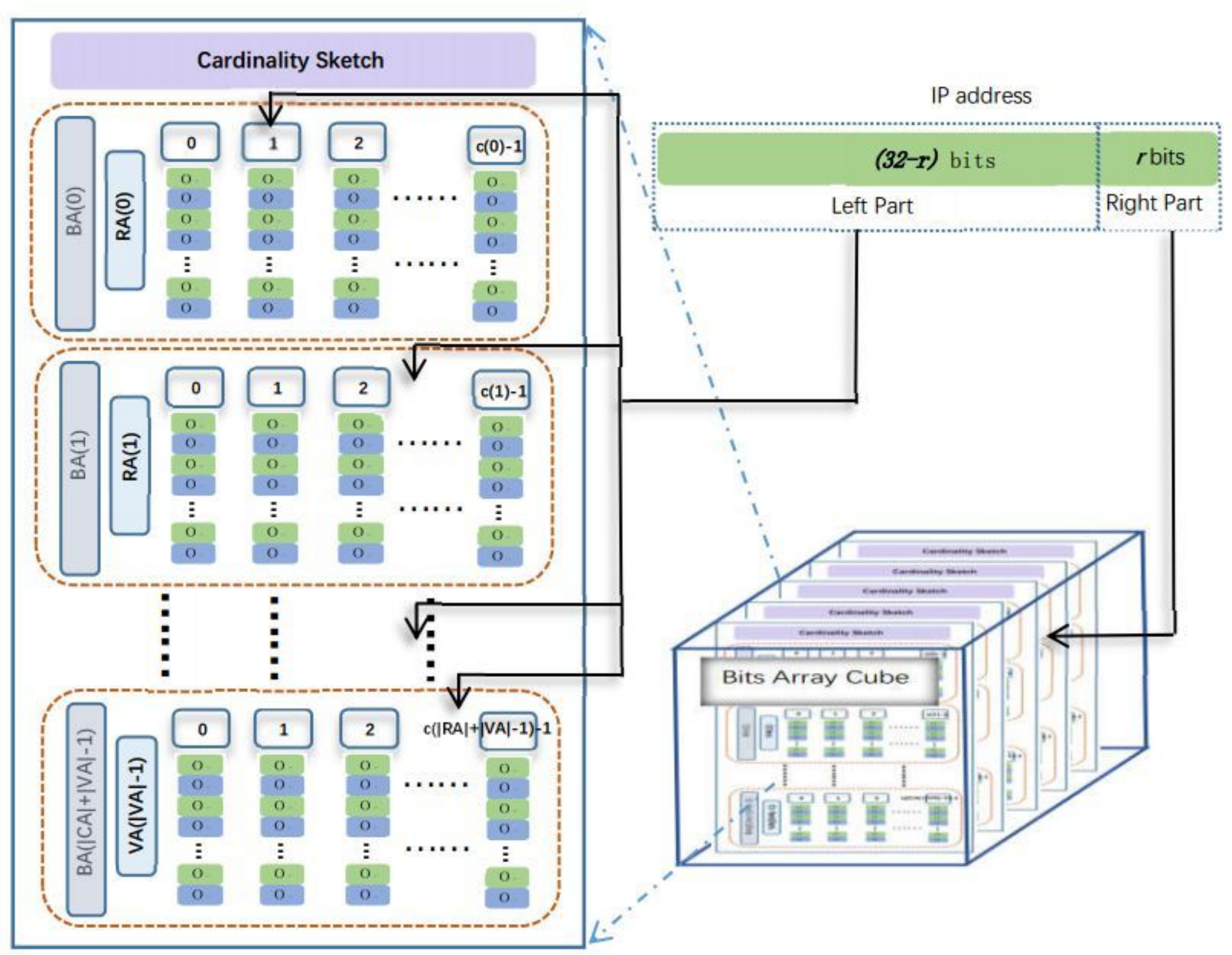}
\caption{Cube of bits array}
\label{BitsArrayCube}
\end{figure}

The inner network may be a big network, such as country network, city network. If we record all these inner IPs cardinalities in a $CS$, the estimating result may deviate the real cardinalities seriously because of the ultra sharing of bits\cite{HSD:identifyHighCardinalityHosts}\cite{HSD:bitmapCountingActiveFlowsHighSpeedLinks}. 

We used several $CS$s to record different inner IPs' cardinalities. The set of these $CS$s is called Cube of Bits Array($CBA$). Using which $CS$ to record an inner IP's cardinality is according the right $r$ bits of this inner IP.
In order to increase the randomness of IP address, each IP will be hashed by a mangling operation\cite{RS2004:ReversibleSketchesEfficientAccurateChangeDetectionNetworkDataStreams}. When restoring super host at the end of a time window, we use the mangling operation again to acquire the origin IP. In the following part, IP means the mangled IP address.
There are $2^r$ $CS$s in $CBA$. Every inner IP could be classified into two parts: Right Part($RP$: the right $r$ bits) and Left Part($LP$: the left $32-r$ bits). $RP$ is used to select $CS$ in $CBA$ and $LP$ is used to determine column index of different $BA$ in the $CS$. An inner IP will select $|RA|+|VA|$ columns from every bit array of a $CS$ to record its cardinality at the same time. $RA$s' columns will be determined by sub bits of $RP$ and $VA$s' columns will be acquired by random hash functions. At the end of a time window, we merge these columns together by ``bits and" operation. The union column will remove some bits set by other inner IPs, but it still contains a little noise. We will remove the resting noise bits by estimating their number.
\begin{definition}[Union Column, UC]
    \label{def-lsb}
   For an inner host $iip$, its union column is the bitwise ``and" result of every column in $CS$ related to it, written as $UC(iip)$.
  \end{definition}

\begin{theorem}
 \label{Th-shareBitsSetProb}
For a $CS$, if there are $\eta$ flows projected to it, the probability that a bit in the union column is set is         $\varepsilon = \prod_{i=0}^{|RA|+|VA|-1}(1-e^{-\frac{\eta}{c(i)*g}})$, where $c(i)$ is the column number of the $i$th bits array.
 \end{theorem}
\begin{proof}
If these $\eta$ flows are project to $BA(i)$ randomly, there will be $\frac{\eta}{c(i)}$ flows in a column. According to equation \ref{BEV_est_value}, the `1' bits of this column is $g-g*e^{-\frac{\eta}{c(i)*g}}$. A bit is set to `1' with probability $P_1(i)=1-e^{-\frac{\eta}{c(i)*g}}$. A bit in the union column is set to 1 with probability $\varepsilon = \prod_{i=0}^{|RA|+|VA|-1}{P_1(i)}=\prod_{i=0}^{|RA|+|VA|-1}(1-e^{\frac{\eta}{c(i)*g}})$. 
\end{proof}

\begin{theorem}
 \label{Th-GetEstValueFromShareCS}
 For an inner host $iip$, its cardinality could be estimated from the zero bits number $Z$ in $UC(iip)$ by equation
 $\widehat{|OP(iip)|}=-g*ln(\frac{Z }{g-g*\varepsilon})$. 
  \end{theorem}
\begin{proof}
 Suppose that there is no other hosts set the bit in $UC(iip)$ except $iip$. According to equation \ref{BEV_est_value}, $\widehat{|OP(iip)|}=-g*ln(\frac{Z_r}{g})$. $Z_r$ is the real zero bits number when the bits in $UC(iip)$ only setted by $iip$.
 
 Then take account on the probability that the $Z_r$ bits are set by other hosts. According to theorem \ref{Th-shareBitsSetProb}, the expecting number of `1' bits is $g-Z_r+Z_r*\varepsilon$. We should remove these `1' bits from $UC(iip)$. In another word, we should add $Z_r*\varepsilon$ zero bits to $UC(iip)$ and $Z_r=Z+Z_r*\varepsilon$. Transforming the above formula, we will get $Z_r=\frac{Z}{1-\varepsilon}$.
 
Put the modified $Z_r$ to equation \ref{BEV_est_value}, we will get $\widehat{|OP(iip)|}=-g*ln(\frac{Z }{g-g*\varepsilon})$.
\end{proof}

Estimate $|OP(iip)|$ by $|UC(iip)|$ maybe get an over-estimating value because of the sharing column with other hosts. So we should add a modified value to $Z$. In the proof processing, $Z_r$ could be rewritten as $Z+\frac{Z*\varepsilon}{1-\varepsilon}$. The additional value $\frac{Z*\varepsilon}{1-\varepsilon}$ could be regarded as the revising value.
 
\subsection{Updating bit array cube}
In a high-speed network, it is burdensome to store packets in memory even only storing IP header. So we can only scan these packets once. When acquiring a packet, we will first extract IP pair like $<iip,oip>$, and update $CBAA$ to record $iip$'s cardinality.

After selecting a $CS$ in $CBAA$ by rightest $r$ bits of $iip$, we should calculate column index of each bit array by left part $LP$ of $iip$. When determining $RA$s' columns indexes, we should make sure that $iip$'s $LP$ can be restored from these indexes. We select some successive bits from $LP$ to generate a column index. Let $CL(i)$ represent the $i$th bit array's column index of $iip$ and $c(i)$ be the total columns number of $i$th bit array. $CL(i)$ can be written in binary format. Supposing $2^{cbn(i)}=c(i)$, $cbn(i)$ is the bits number of $CL(i)$. $cbn(i)$ is also the number of successive bits we want to extract from $LP$. Together with the start bits $CL_{bs}(i)$, we can map $iip$ to a unique column in $i$th bit array. $CL_{bs}(i)$ is the start bits offset and  $r \leq CL_{bs}(i)\leq 31$. Fig.\ref{SubIPtoColIdx} shows how to determine each restoring arrays' columns by $iip$'s left part.

\begin{figure}[!ht]
\centering 
\includegraphics[width=0.47\textwidth]{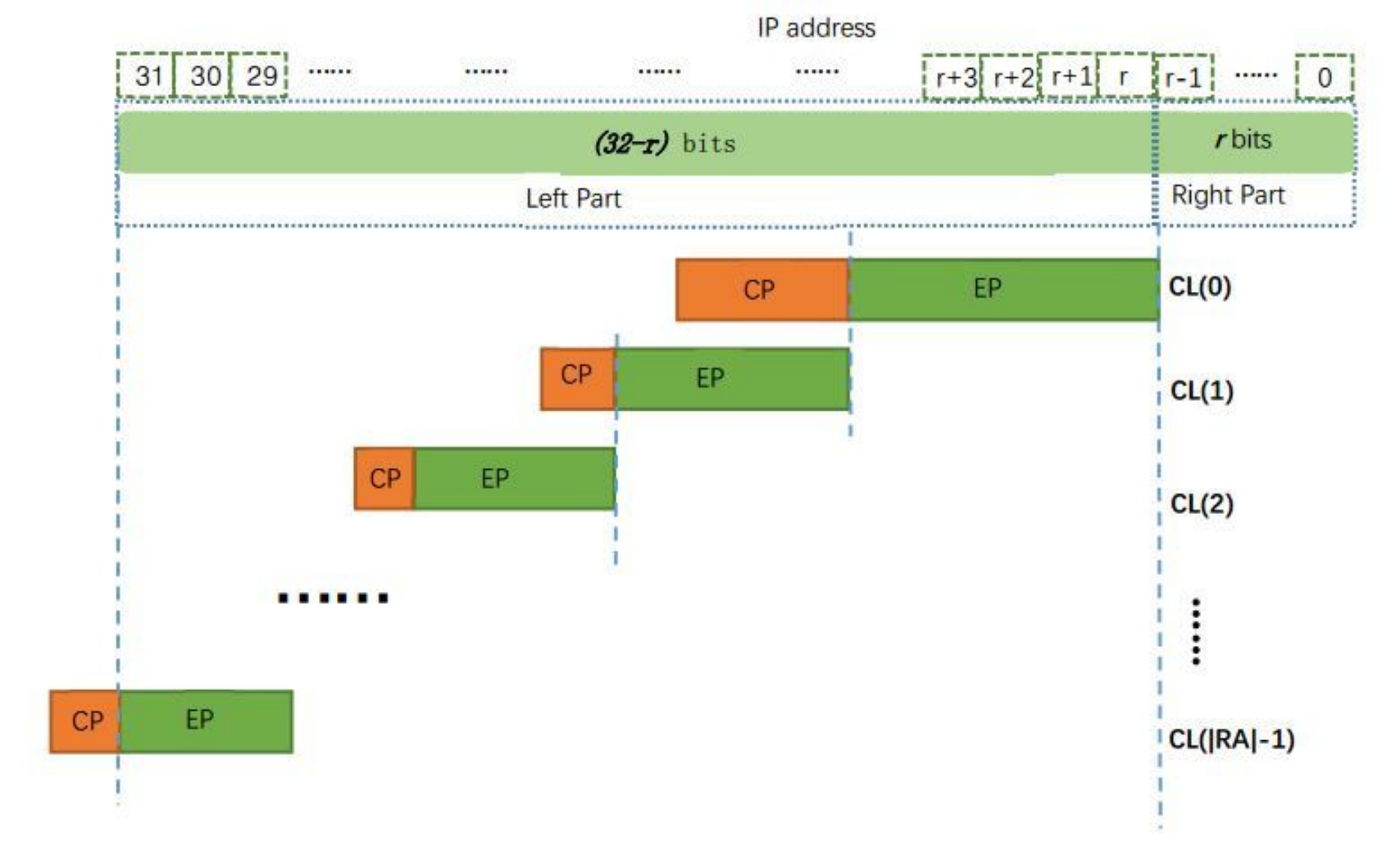}
\caption{Select columns of recover arrays}
\label{SubIPtoColIdx}
\end{figure}

We divide $CL(i)$ into two parts logically: Efficient Part (EP) and Checking Part (CP). Let $EP(i)$ and $CP(i)$ represent efficient part and checking part of the $i$th array's column. Denote $EP(i)$ and $|CP(i)|$ as the length of $EP(i)$ and $CP(i)$. Each column's EP is pairwise disjoint and their union is LP. CP is used to filter fake columns tuple when restoring super hosts. $CP(i)$ equals to the first $|CP(i)|$ bits of $EP(i+1)$, where $0\leq i \leq |RA|-1$ and when $i= |RA|-1$, $i+1=0$.  A longer CP is beneficial to super hosts restoring but occupy more memory too. CP of the last RA's column relates to bits with index more than 31. In the case, we module this bits' indexes with base 32.

Indexes of VAs' columns are selected by random hash functions $H_i$, where $H_i$ is a hash function that maps a value between $[0,2^{32}-1]$ to a value between $[0,c(i)-1]$.
 
Algorithm \ref{updateIPpair} describes IP pair scanning process in each local server.          
\begin{algorithm}                       
\caption{Update IP addresses pair}          
\label{updateIPpair}                            
\begin{algorithmic}                    
    \Require {\\ IP address pair $<iip,oip>$,\\
     CBAA}\\
     
\State $bvIdx \Leftarrow H_{bv}(oip)$
\State $csIdx \Leftarrow $ right $r$ bits of $iip$
\State $pCS \rightarrow $ $csIdx$th CS in CBAA
\State $LP \Leftarrow$ left (32-r) bits of $iip$
\For $RA(i)$ in $pCS$
\State $CL(i) \Leftarrow $ extract $cbn(i)$ bits from $CL_{bs}(i)$ in $iip$
\State set $bvIdx$th bit in column $CL(i)$ to 1
\EndFor
\For $VA(j-|RA|)$ in $pCS$, $|RA|\leq j \leq |RA|+|VA|-1$
\State $CL(j) \Leftarrow H_j(LP)$
\State set $bvIdx$th bits in column $CL(j)$ to 1
\EndFor

\end{algorithmic}
\end{algorithm}
Each local server will run algorithm \ref{updateIPpair} to scan packets. When an IP pair updates the $CBAA$, only $|RA|+|VA|$ bits will be changed to 1 without reading any bytes from $CBAA$. There is no reading related conflicts during the updating process and algorithm \ref{updateIPpair} matches the Bernstein's conditions \cite{PL:BerensteinAnalysisProgramsParallelProcessing} \cite{PL:SharedMemorySynchronization} which ensures that the updating process can run parallel.  After scanning all the packets in a time window, a recovering algorithm will be applied to find out super hosts from these local servers.

\subsection{Restore super hosts on global server}
At the end of a time window, each local server will send its $CBA$ to the global server. $CBA$ on each local server has the same structure: the same number of $CS$, the same columns number in each bit array. Global server merges these $CBA$ by ``bits OR" operation. The merged $CBA$ contains cardinalities information of all hosts in the inner network. We can restore super hosts from the global $CBA$. 

$RP$ of an IP address could be acquired by $CS$'s id in $CBA$. The task of super hosts restoring is to get $LP$ from a $CS$. In each bit array, there is a kind of columns called ``hot columns" as defined below.

\begin{definition}[Hot Column, HC]
    \label{def-hotColumn}
For a column in a bit array, if there is one or more super hosts projecting to it, this column is called hot column.
  \end{definition}

IP's left part hides in $RA$'s columns. As shown in fig.\ref{SubIPtoColIdx}, if we get $iip$'s columns indexes of $RA$s, we can concatenate their $EP$s together. But we don't know $iip$'s columns indexes when restoring super hosts. Fortunately, we can know which columns are hot columns by calculating its zero bits number. By testing tuple of these kinds of columns of each $RA$, we can recover $iip$'s $LP$ indirectly. Generally speaking, there are two steps in restoring $LP$: hot columns calculation and columns tuple checking.

According to theorem \ref{Th-GetEstValueFromShareCS}, if a host's cardinality is no less than $\theta$, there will be less than $\theta_{bn}=g(1+\varepsilon)*e^{-\frac{\theta}{g}}-g*\varepsilon$ zero bits in these columns relating to it. Scan each column of $RA$ and record these columns indexes whose zero bits number is no more than $\theta_{bn}$ to a list $HC(i)$. $HC(i)$ stores $i$th $RA$'s hot columns indexes. Algorithm \ref{alg-getHotColumn} describes how to mining hot columns from $CS$.

\begin{algorithm}                       
\caption{Locating hot columns}          
\label{alg-getHotColumn}                            
\begin{algorithmic}             
\Require {\\ $\theta_{bn}$,\\ cardinality sketch $CS$;} 
\Ensure {\\Hot columns lists set\\ $HC=\{HC(0),HC(1),\cdots,HC(|RA|-1)\}$;\\}

\For{ each restoring array $RA(i)\in CS$ }
   \For{ $j \in [0,c(i)]$ }
   \If {there are no more than $\theta_{bn}$ zero bits in $j$th column of $RA(i)$}
    \State Insert $j$ into $HC(i)$
   \EndIf
   \EndFor
   \State Insert $HC(i)$ into $HC$
\EndFor
\State Return $HC$
\end{algorithmic}
\end{algorithm}

super hosts' $LP$s are hiding in these hot columns' efficient parts.
Selecting $|RA|$ hot columns from every $HC(i)$, we will get a hot columns tuple $THC=<hc_0,hc_1,\cdots,hc_{|RA|-1}>$ where $hc_i\in HC(i)$. There are total $|THC|=$ $\prod_{i=0}^{|RA|-1}{|HC(i)|}$ different such tuples. Because we don't know which tuple contains super hosts, we had to test them one by one. Hot columns tuple checking contains two parts: checking if this tuple can restoring $LP$ and testing if the union of columns related to the restoring $LP$ is still a hot column.

According to the relationship between IP's different columns, we can filter most fake tuple by comparing each column index's last $|CP(i)|$ bits with next column index's first $|CP(i)|$ bits. Now, we only know the offset and length of each column index. The next question is how to get the length of different columns $CP$. $EP$ and $CP$ of different columns are divided logically. For the $i$th restoring bit array, the $EP$ of its column's index is these bits that appear in this bit array's column index, not in $(i+1)$th bit array's. So we can get the length of $EP(i)$ by calculating the difference of offset between contiguous bit arrays' column indexes. In another word, $|EP(i)|=CL_{bs}((i+1)mod(|RA|))-CL_{bs}(i)$, where $0\leq i \leq |RA|-1$. Because $CP$ is the resting part of $LP$, $|CP(i)|=cbn(i)-|EP(i)|$. If all columns in a tuple are corresponding to this condition, we can restore $LP$ from each column's $EP$. Algorithm \ref{alg-checkColTuple} illustrates how to check each columns tuple.

\begin{algorithm}                       
\caption{Checking hot columns tuple}          
\label{alg-checkColTuple}                            
\begin{algorithmic}             
\Require {\\ Hot columns tuple$THC=<hc_0,hc_1,\cdots,hc_{|RA|-1}>$ ,\\ $\theta_{bn}$ ,\\ cardinality sketch $CS$;} 
\Ensure {\\ super host's $LP$;\\}

\State $LP \Leftarrow -1$
\For{ $i \in [0,|RA|-1]$ }
    \State $|CP(i)| \Leftarrow CL(i)-|EP(i)|$
   \If {$CP$ of $hc_i$ not equal to the first $|CP(i)|$ bits of $hc_{(i+1)mod(|RA|)}$}
    \State Return $LP$
   \EndIf
\EndFor
\State $LP \Leftarrow $ concatenation of all columns' $EP$
\State $UCol \Leftarrow $ $hc_0$th column of $RA(0)$
\For {$ i \in [1,|RA|-1]$}
\State $UCol$ ``AND" with $hc_i$th column of $RA(i)$
\EndFor
\For {$ j \in [0,|VA|-1]$}
  \State $UCol$ ``AND" with $H_j(LP)$th column of $VA(j)$
\EndFor
\If {zero bits number of $UCol$ bigger than $\theta_{bn}$}
\State $LP \Leftarrow -1$
\EndIf
\State Return $LP$
\end{algorithmic}
\end{algorithm}

When using algorithm \ref{alg-checkColTuple} to check a columns tuple, if columns can not generate a $LP$ of some IP address or no super hosts relating to this columns tuple, this algorithm will return a negative value -1, otherwise it will return the restoring $LP$. By concatenating restored $LP$ with $RP$, which is the id of $CS$, we get the whole IP address. No matter scanning packets or checking hot columns tuple, we can run them parallel in order to acquire a high processing speed. 

\section{Implement on GPU}
The use of graphics processing units (GPUs) becomes a significant advance to speed up the packets scanning by taking the advantage of the massive parallelism capabilities of GPUs. With the acceleration of GPU, thousands of packets can be disposed at the same time. 

OpenCL and CUDA are two famous GPGPU models. OpenCL, established by the Khronos Group \cite{OpenCL:group}, is a framework for writing parallel programs that execute across heterogeneous platforms consisting of CPUs, GPUs, and other processors. As a result, OpenCL provides software developers with portable and efficient access to the power of diverse processing platforms \cite{GPU:SystemCSimulationOnGP-GPUsCUDAOpenCL} \cite{OpenCL:AParallelProgrammingStandardForHeterogeneousComputingSystems}. Unlike OpenCL, CUDA \cite{CUDA:nvidia} only supports GPUs produced by NVIDIA. But it has the best performance for NVIDIA's GPUs. 

With the help of NVIDIA's Compute Unified Device Architecture (CUDA), deploying applications on GPU becomes easier than before when graphics language were required if someone wanted to run other program on GPU. In the CUDA execution model, threads running on GPU device follows the single instruction multiple data (SIMT) model of execution. Under SIMT modal, each packet could be viewed as a data to be processed with the same instructions set by a thread. In CUDA model, millions of threads could be started to deal with huge packets. 

Taking advantage of the full programmability offered by CUDA and the potential parallel ability of GSDA, we can detect super hosts in real time.
 
\subsection{Update packets on GPU}
 
Each IP pair could be regarded as a datum unit to be processed by a thread. Thousands of threads could be started on GPU at the same time to handle thousands of IP pairs parallel.  
Fig.\ref{GPU_SuperhostsDectionModel} demonstrates how to detect super hosts on the local server with GPU accelerating.

All IP pairs will be firstly buffered on local server's memory and then transmitted to GPU's global memory for cardinalities recording.
\begin{figure}[!ht]
\centering
\includegraphics[width=0.47\textwidth]{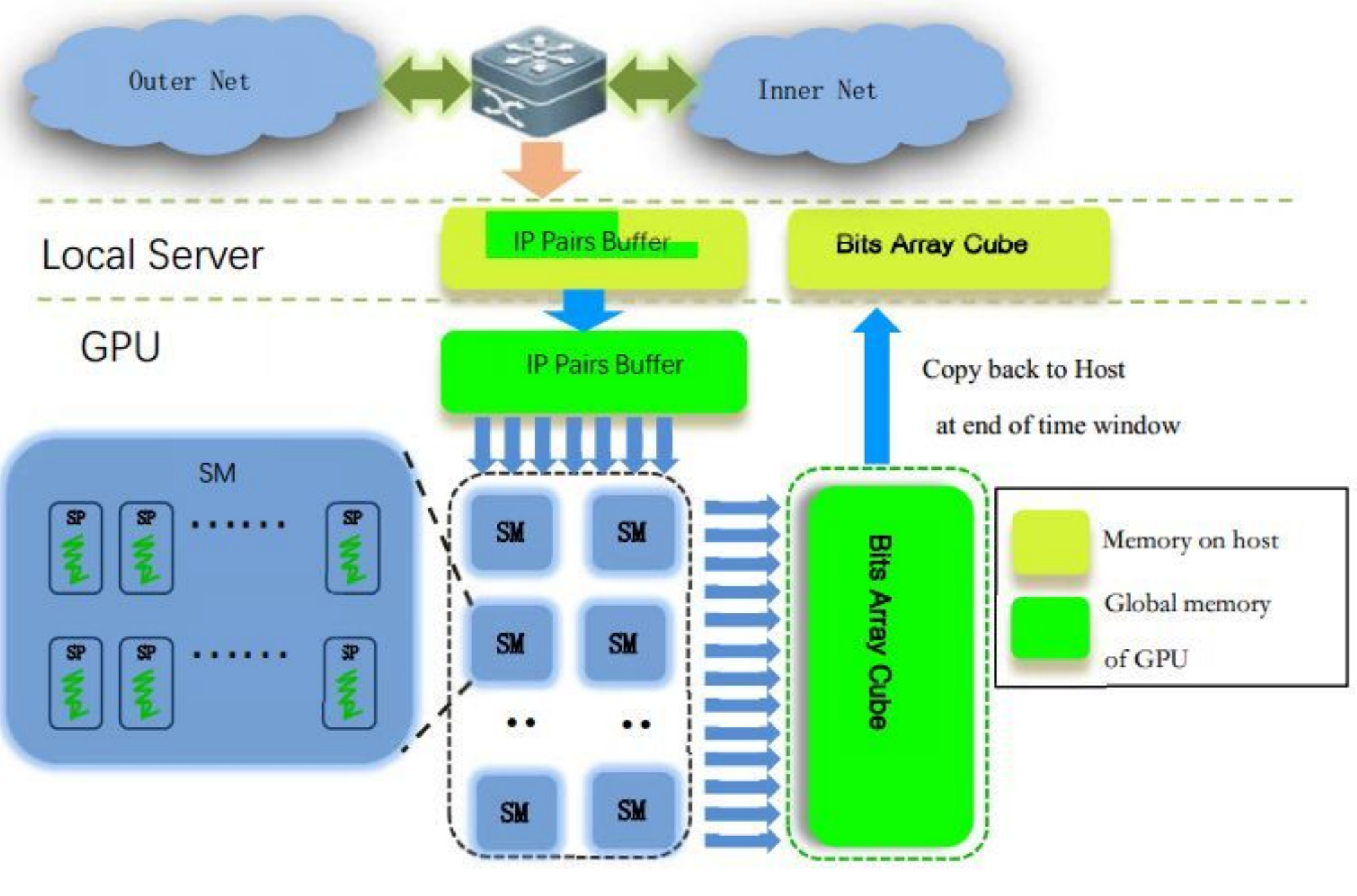}
\caption{Super hosts detection on GPU}
\label{GPU_SuperhostsDectionModel}
\end{figure}

Every thread will run the algorithm \ref{updateIPpair} to handle an IP pair. These threads will read different IP pairs parallel and update $CBA$.

At last, when all these IP pairs in a time window are scanned, $CBA$ on local server will be sent to global server and merged together for super hosts detection. 

\subsection{Restoring Super hosts on GPU}
At the restoring stage, there are three parts could be running parallel on GPU: merging local $CBA$, locating hot columns of different bit array in global $CBA$ and checking hot columns tuple.

After receiving local $CBA$s from different local servers, the global server will merge them together one by one for global super host detection. We will first initialize a $CBA_g$ which has the same structure as local server's $CBA$ has. Every bit of this initial $CBA_g$ on the global server is set to 1. Every local $CBA$ will be merged to $CBA_g$ by bitwise ``AND" operation. There is no data conflict between columns of $CBA$, we can merge different columns by several GPU threads concurrently.

Every column of different bit arrays has the same rows (bits vector contains the same number of bits). Using a GPU kernel thread to checking if a column is a hot column by counting the `0' bits in it. GPU can launch thousands of such threads for hot column checking parallel. The number of columns is fixed. Both global $CBA$ merging and hot columns checking will not consume much time on GPU. 

The complexity of columns tuple checking depends on the number of super hosts. The more super hosts, the more candidate columns tuple. In algorithm \ref{alg-checkColTuple}, we need to use a buffer with size $\frac{g}{8}$ bytes for union column. When a thread checking a columns tuple, we may allocate this buffer temporarily and free it before this thread return. This method is flexible for threads evoking. But it will waste time for GPU memory allocating and freeing. What's more, when there are too many threads, memory allocating may cause an error because of too many memory requirements. In order to speed up columns tuple checking speed and avoid additional running error, we allocate a buffer pool for all tuple checking threads on GPU. A buffer in the pool is used by only a thread once a time. With the help of buffers pool, every thread just needs to set every bit of the buffer to 1 at the begin. Although the number of buffers in the pool limits the number thread running parallel, we may allocate a buffers pool as large as possible.

$CBA$ is a smart structure. There is no data conflict when updating or reading $CBA$. When deploying our algorithm on GPU, we can detect high-speed networks' super hosts in real time.

\section{Experiment and analysis}
We use real-world backbone traffic to evaluate the performance of different super hosts detection algorithms. All of these algorithms run on a PC with Nvidia GPU(NVIDIA Titan Xp, with 12 GB global memory).  
\subsection{Experiment data}
The network traffic used in our experiment is acquired from CERNET \cite{expdata:IPtraceCernetJS}. In the experiments, we use 1 hour traffic starting from 13:00 to 14:00 On October 23, 2017. The size of time window is 5 minutes. There are 12 time windows in the traffic. We used different algorithms to mining super hosts of every time window. Table \ref{tbl-trafficInf} describes the traffic information of each time window.

\begin{table*}
\centering
\caption{Experiment Data}
\label{tbl-trafficInf}
\begin{tabular}{c}                                                                                                                                                                                                                           
\centering
\includegraphics[width=0.7\textwidth]{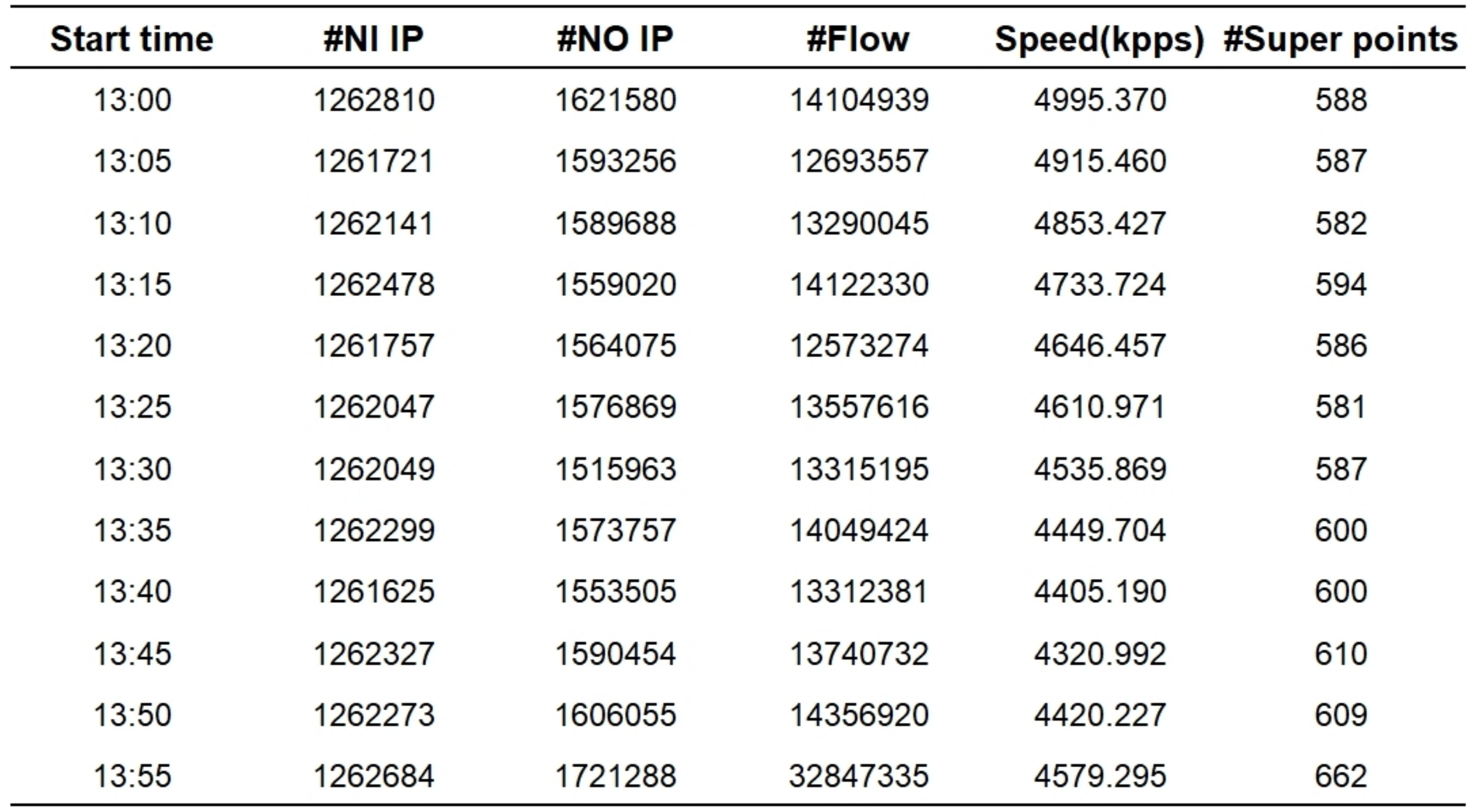}
\end{tabular}
\end{table*}

In our experiments, super host threshold $\theta$ is set to 1024. All super hosts of different traffic datasets in different time window were acquired offline precisely. In order to get these baselines, we use ``map container" of C++ STL to store every hosts' opposite points and sort every host by their cardinalities descent. Every algorithm's detection result will be checked by these baselines to evaluate accuracy.

\subsection{Algorithm accuracy}
We compared the rightness of our algorithm CBAA with DCDSA\cite{HSD:ADataStreamingMethodMonitorHostConnectionDegreeHighSpeed}, VBFA \cite{HSD:DetectionSuperpointsVectorBloomFilter} and GSE \cite{HSD:GPU:2014:AGrandSpreadEstimatorUsingGPU}. False Negative Ratio (FNR) and False Positive Ratio (FPR) defined by equation \eqref{eq-fnr} and \eqref{eq-fpr} are adopted to evaluate the detection accuracy of these algorithms. 
 \begin{equation}\label{eq-fnr}
FNR=\frac{\left \|\hat{H}^-\right \|}{\left \|H\right \|}
\end{equation} 
\begin{equation}\label{eq-fpr}
FPR=\frac{\left \|\hat{H}^+\right \|}{\left \|H\right \|}
\end{equation} 
Let $\hat{H}$ represent the set of detected super hosts by estimating algorithms. H is the set of super hosts whose cardinality is no less than the threshold. $H,\hat{H}^-,\hat{H}^+$ are defined as below.
$$H=\left \{h| cardinality\ of\ h \geq \theta \right \}$$
$$\hat{H}^-=\left\{h|h \in H,h \notin \hat{H}\right \}$$
$$\hat{H}^+=\left\{h|cardinality\ of\ h \leq {\theta },h \in \hat{H}\right \}$$
The sum of FPR and FNR measures the total false rate of an algorithm, written as FTR.

Fig.\ref{fig_exp_rlt_compare_fpr}, fig.\ref{fig_exp_rlt_compare_fnr} and fig.\ref{fig_exp_rlt_compare_ftr}show different algorithms' FPR, FNR and FTR.

\begin{figure*}[!ht]
\centering
\includegraphics[width=0.9\textwidth]{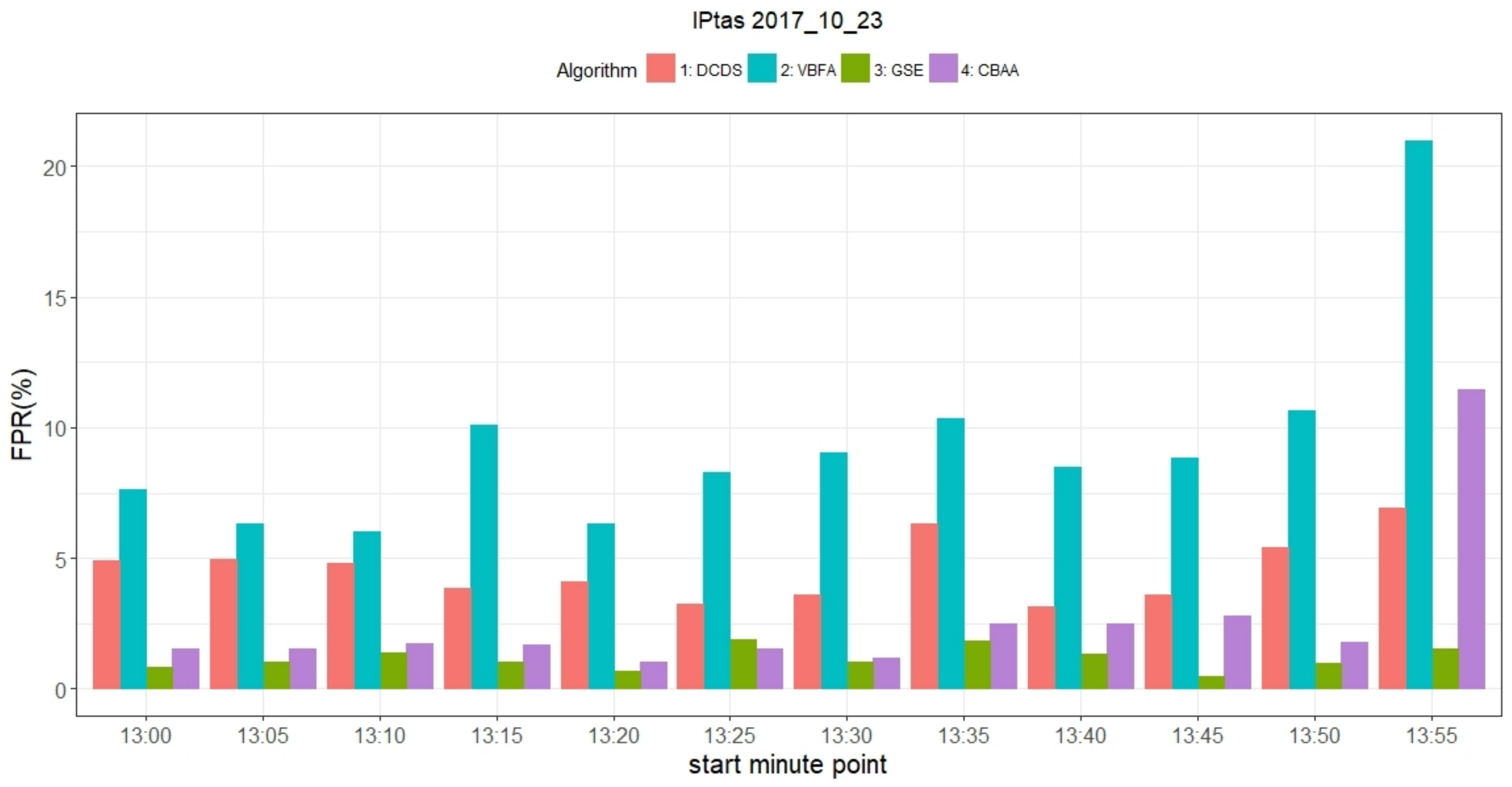}
\caption{False positive rate of different algorithms}
\label{fig_exp_rlt_compare_fpr}
\end{figure*}

\begin{figure*}[!ht]
\centering
\includegraphics[width=0.9\textwidth]{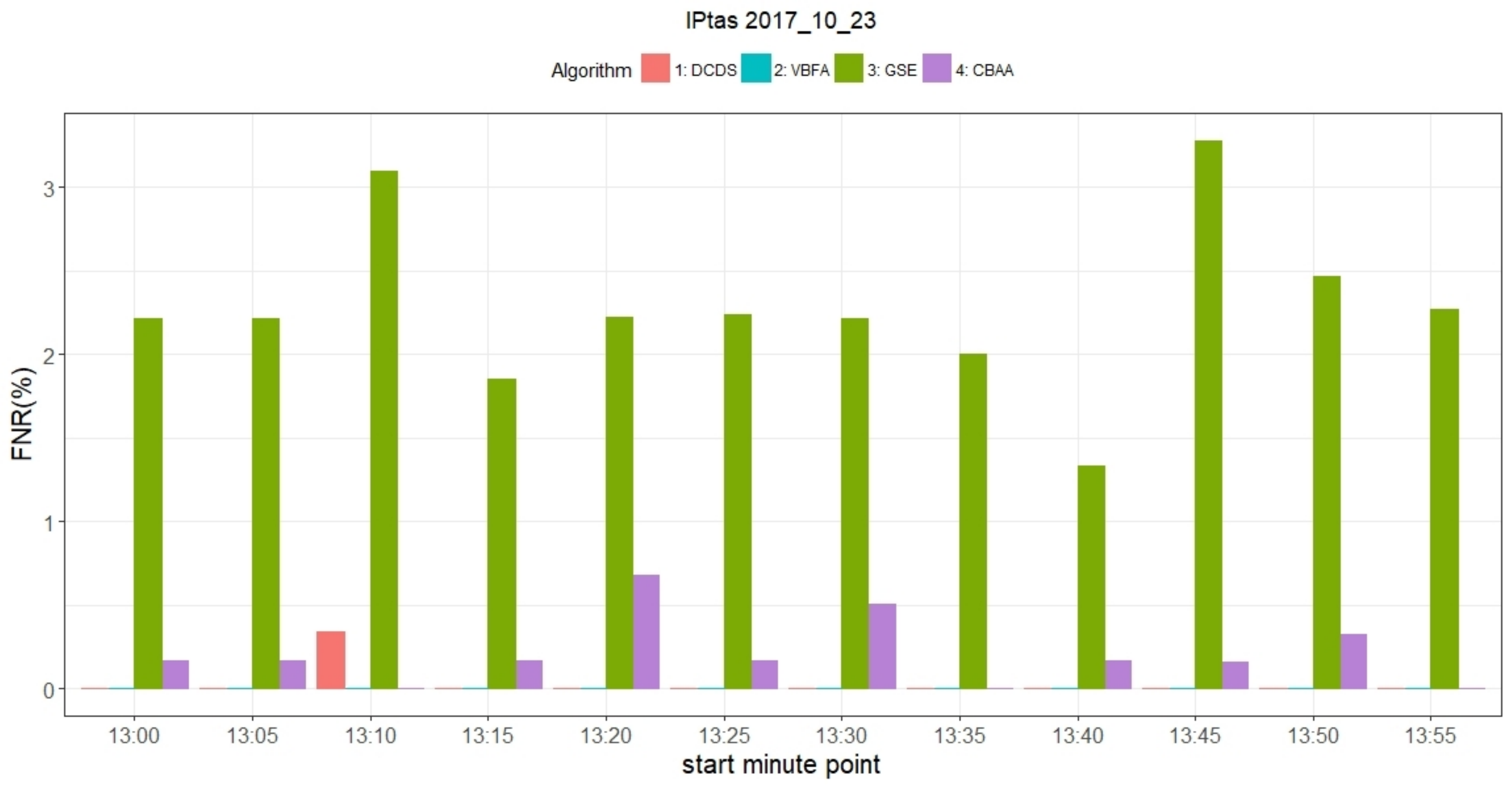}
\caption{False negative rate of different algorithms}
\label{fig_exp_rlt_compare_fnr}
\end{figure*}

\begin{figure*}[!ht]
\centering
\includegraphics[width=0.9\textwidth]{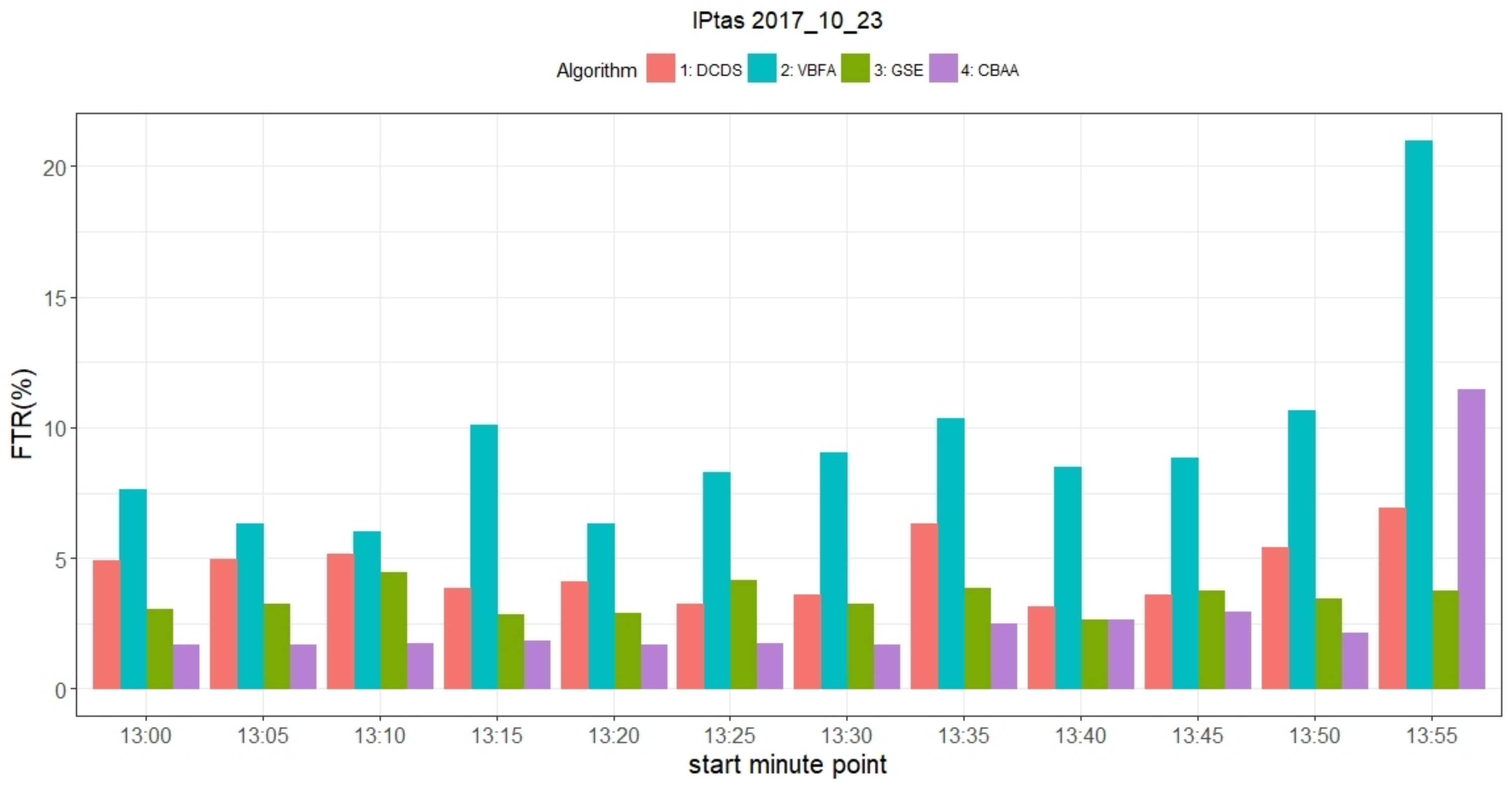}
\caption{False total rate of different algorithms}
\label{fig_exp_rlt_compare_ftr}
\end{figure*}

In our experiments, VBFA has the lowest FNR, but its FPR is the highest. VBFA has the highest FTR, and its average FTR is as high as $9.419\%$. DCDS also has a low FNR and a high FPR. But the FPR of DCDS is smaller than that of VBFA. Compared with VBFA, GSE has the lowest FPR and the highest FNR. The average FTR of GSE is $3.449\%$. 

The average FNR of CBAA is $0.211\%$ which is a little higher than that of DCDS and VBFA. But the average FTR of CBAA is the lowest which is only $2.82\%$.

\subsection{Using time and memory}
CBAA has not only a high accuracy but also the fastest super hosts restoring speed. At the end of each time window, CBAA restores super hosts from each cardinality sketch. Each cardinality sketch only contains a fraction of super hosts. There are fewer candidate super hosts in each cardinality sketch which makes sure that CBAA has the fastest super hosts restoring speed. Fig.\ref{fig_exp_rlt_compare_rstSPt} illustrates time consumption of different algorithms.
\begin{figure*}[!ht]
\centering
\includegraphics[width=0.9\textwidth]{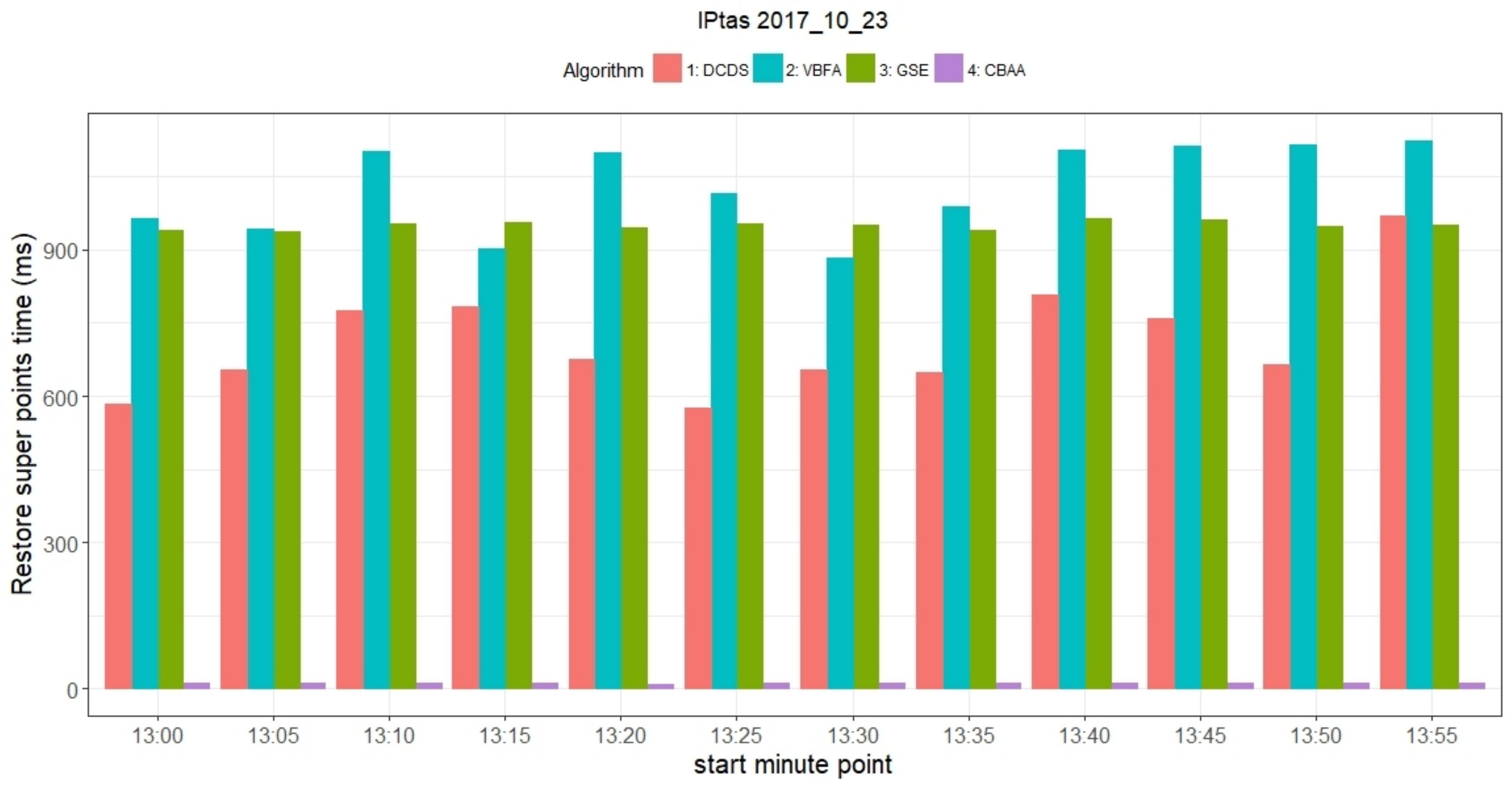}
\caption{Processing time of different algorithms}
\label{fig_exp_rlt_compare_rstSPt}
\end{figure*}

VBFA's time consumption is affected by super hosts number. When flow number and super hosts number are big, VBFA will generate and test huge candidate IPs. This will waste much time. In our experiment data, there are 599 super hosts in each time window on average. VBFA consumes average 1029.7 milliseconds to restore super hosts in each time window. DCDS adopts Chinese Remainder Theorem, which needs complex computation, to restore super hosts. Its speed is also lower. GSE needs to estimate the cardinality of each host to detect super hosts. Its super host restoring time is also very high. Our novel algorithm CBAA can mine super hosts with less than 11 milliseconds on average. With this fast speed, CBAA can be deployed on backbone network with a cheap GPU device.

In our experiments,DCDS consumed 384 MB of memory(H = 2,H+ = 1,m = $2^{16}$,v = $2^{14}$\cite{HSD:ADataStreamingMethodMonitorHostConnectionDegreeHighSpeed}), VBFA used  160 MB of memory(m = $2^{13}$\cite{HSD:DetectionSuperpointsVectorBloomFilter}) and GSE used 128 MB of memory(z=$2^{20}$,
 m=$2^{30}$
 \cite{HSD:GPU:2014:AGrandSpreadEstimatorUsingGPU}). For CBAA, we set r=4, $|RA|=3$, $|VA|=1$, $g=2^{12}$,$cbn(0)=cbn(1)=cbn(2)=cbn(3)=2^{12}$ and the memory consumption is 128 MB. CBAA has the highest accuracy and fastest speed with a small memory consumption. It has better performance than other algorithms.

\section{Conclusion}
We devise a fast super host detection algorithm, CBAA, for high-speed networks. No existing super hosts detection algorithm can acquire this accuracy and speed because of their complex super hosts restoring process. CBAA is suitable for distributed deploying and can be efficiently implemented on a GPU card. This advantage comes from the smart cardinality recording structure proposed in this paper. According to our experiments on a core network traffic, our algorithm can detect super hosts with the fastest speed and the highest accuracy among these state-of-the-art algorithms.

\section*{Acknowledgment}
The authors would like to thank anonymous reviewers. The research work leading to this article is supported by the National Natural Science Foundation of China under Grant No. 61602114

\iftoggle{ACM}{
\bibliographystyle{ACM-Reference-Format}
}
\iftoggle{IEEEcls}{
\bibliographystyle{IEEEtran}
}

\bibliography{..//..//ref} 

\end{document}